\documentclass[conference]{IEEEtran}
\IEEEoverridecommandlockouts
\usepackage{graphicx} 
\usepackage{cite}     
\usepackage{amsmath}  
\usepackage{float}
\usepackage{stfloats}
\usepackage{amsthm}

\newtheorem{lemma}{Lemma}

\newtheorem{remark}{Remark}
\usepackage{amsfonts} 
\usepackage{amssymb}  
\usepackage{algorithm}
\usepackage{algpseudocode}
\usepackage{multirow}
\usepackage{color}
\usepackage{comment}
\usepackage[none]{hyphenat}
\usepackage{balance}

\begin{document}

\title{Clique-Based Deletion-Correcting Codes via Penalty-Guided Clique Search}

\author{\IEEEauthorblockN{Aniruddh Pandav and Rajshekhar V Bhat\\}
\IEEEauthorblockA{Indian Institute of Technology Dharwad, Dharwad, India\\
Email: \{ee23bt031, rajshekhar.bhat\}@iitdh.ac.in}}

\maketitle

\begin{abstract}
We study the construction of $d$-deletion-correcting binary codes by formulating the problem as a Maximum Clique Problem (MCP). In this formulation, vertices represent candidate codewords and edges connect pairs whose longest common subsequence (LCS) distance guarantees correction of up to $d$ deletions. A valid codebook corresponds to a clique in the resulting graph, and finding the largest codebook is equivalent to identifying a maximum clique. While MCP-based formulations for deletion-correcting codes have previously been explored, we demonstrate that applying \emph{Penalty-Guided Clique Search (PGCS)}, a lightweight stochastic clique-search heuristic inspired by Dynamic Local Search (DLS), consistently yields larger codebooks than existing graph-based heuristics, including minimum-degree and coloring methods, for block lengths $n \in \{8,9,\ldots,14\}$ and deletion parameters $d \in \{1,2,3\}$. In several finite-length regimes, the resulting codebooks match known optimal sizes and outperform classical constructions such as Helberg codes. For decoding under segmented reception, where codeword boundaries are known, we propose an optimized LCS-based decoder that exploits symbol-count filtering and early termination to substantially reduce the number of LCS evaluations while preserving exact decoding guarantees. These optimizations lead to significantly lower average-case decoding complexity than the baseline
$O(|\mathcal{C}|n^2)$ approach.
\end{abstract}

\section{Introduction}
\label{Introduction}
Deletion channels, where transmitted symbols are removed at unknown locations, are increasingly relevant in modern communication systems. In high-rate transmissions, precise symbol timing becomes critical as symbol durations shrink; synchronization errors or timing drifts can cause receivers to miss symbols, effectively leading to deletions \cite{symbol_sync}. Similar effects arise in other contexts such as wireless packet networks, molecular communication, and data storage \cite{DNA_storage}. Unlike substitution or erasure, deletions disrupt sequence alignment between the transmitter and receiver, complicating both code design and decoding.

The goal of this work is to construct deletion-correcting codebooks that remain robust against up to $d$ deletions in length-$n$ binary sequences. We formulate this construction as a Maximum Clique Problem (MCP), following the approach first explored in~\cite{Compare}. In this formulation, a graph is considered in which each vertex represents a binary sequence of length $n$, making it a candidate codeword from the $2^n$ possible sequences. An edge connects two vertices if the corresponding codewords are mutually correctable; specifically, if their longest common subsequence (LCS) distance is sufficiently large such that they remain distinguishable after up to $d$ deletions. Each clique in this graph (a subgraph in which every pair of vertices is connected) corresponds to a valid codebook, and finding the largest feasible codebook is equivalent to identifying a maximum clique.

Prior work on deletion-correcting codes has explored both algebraic and heuristic methods. The Varshamov--Tenengolts (VT) codes~\cite{VT_original} achieve optimal single-deletion correction with $\log_2(n+1)$ bits of redundancy. Efficient systematic encoding algorithms and non-binary extensions of VT codes have also been studied in~\cite{VT_code}. Extensions to multiple deletions include Helberg codes~\cite{Helberg} and their non-binary generalizations~\cite{LeNyugen}. Graph-based formulations have also emerged: a graph coloring approach~\cite{coloring} proved VT codes optimal in a coloring sense and demonstrated constructions from constant-weight subgraphs, while a maximal independent set heuristic~\cite{Compare} was proposed for finding deletion-correcting codes.

In contrast to earlier graph-based heuristics, we employ
\emph{Penalty-Guided Clique Search (PGCS)}, a lightweight stochastic heuristic
for finding large cliques in deletion-compatibility graphs. PGCS employs a
dynamic penalty mechanism inspired by Dynamic Local Search (DLS)~\cite{DLS},
combined with greedy clique expansion and penalty-driven perturbations, while
avoiding exhaustive enumeration and deterministic branch-and-bound strategies. A key advantage of PGCS is its explicit control over the computational budget:
the algorithm can be configured to run for a fixed number of iterations or up to
a prescribed time limit, enabling predictable and resource-efficient execution.
This property makes PGCS particularly attractive for moderate-scale code
construction tasks, where exact graph-search methods, such as the
Bron--Kerbosch algorithm, quickly become computationally prohibitive. To the
best of our knowledge, this work represents the first application of a
penalty-guided stochastic clique-search heuristic inspired by DLS to the
construction of deletion-correcting codes. Our experimental results demonstrate
that PGCS consistently discovers larger cliques (i.e., codebooks) than existing
heuristic methods for given $(n,d)$ parameters.

For decoding, we build upon LCS-based approaches demonstrated in DNA storage reconstruction algorithms~\cite{DNA_sequence}. Previous studies on segmented error-correction~\cite{segmented_decoding, DNA_sync, Marker_code_marker,DNA_deletion} have shown that leveraging known codeword segmentation and synchronization markers simplifies the decoding process and reduces alignment ambiguity. Inspired by these works, we refine the LCS-based decoding framework by incorporating a lightweight candidate elimination step guided by simple codeword statistics. This optimization substantially reduces computational and memory overheads without compromising reconstruction accuracy, under the assumption that the receiver processes each codeword independently.

The main contributions of this work are as follows:
\begin{itemize}
 \item We propose  PGCS, a stochastic heuristic inspired by DLS~\cite{DLS} that combines
greedy clique expansion with lightweight penalty-driven perturbations to
efficiently find large cliques, yielding larger codebooks than existing
clique-based construction methods.

    \item We propose an optimized LCS-based decoding algorithm that computes the longest common subsequence between the received sequence and candidate codewords. Unlike brute-force approaches that enumerate all  $\binom{n}{d}$ deletion patterns, LCS-based decoding operates in $O(|\mathcal{C}| n^2)$ time, where $|\mathcal{C}|$ is the cardinality of the codebook. Our method further reduces this cost by first eliminating unlikely codewords using lightweight symbol-level checks, thereby decreasing the number of required LCS computations without compromising decoding accuracy.

\item Our experimental results demonstrate that the proposed PGCS-based
construction produces larger codebooks than existing clique-based methods
across block lengths $n \in \{8, 9, \ldots, 14\}$ and deletion capacities
$d \in \{1, 2, 3\}$, and matches provably optimal codebook sizes in some
instances.

\end{itemize}

\section{Problem Formulation}
\label{Problem}
Let $n$ be a positive integer and $\mathcal{S}=\{0,1\}^n$ the set of all length-$n$ binary sequences.
The deletion channel with at most $d$ deletions maps any transmitted codeword
$u\in\mathcal{S}$ to a subsequence of length $n-t$ for some $t$ with $0\le t\le d$.
For each $u\in\mathcal{S}$, define the set of possible received subsequences
\begin{multline}
\mathcal{D}_d(u)
= \bigcup_{t=0}^{d}
\Bigl\{
y\in\{0,1\}^{\,n-t}:\;
y\text{ is a subsequence of }u\\
\text{formed by deleting }t\text{ symbols}
\Bigr\}.
\end{multline}

A \emph{codebook} $\mathcal{C}\subseteq\mathcal{S}$ of size $M = |\mathcal{C}|$ is paired with an \emph{encoder} 
\begin{equation}
\psi: \{1, 2, \ldots, M\} \to \mathcal{C}
\end{equation}
that maps each message index to a distinct codeword. We require $\mathcal{C}$ to be \emph{$d$-deletion-correcting}, i.e., no received sequence may arise from two distinct codewords:
\begin{equation}
\label{deletion_correction}
\mathcal{D}_d(u)\cap\mathcal{D}_d(v)=\varnothing
\quad\forall\; u\neq v\in\mathcal{C}.
\end{equation}

Let $\mathcal{Y} = \bigcup_{u\in\mathcal{C}} \mathcal{D}_d(u)$ denote the set of all possible received sequences.
A \emph{decoder} is a mapping
\begin{equation}
\phi:\mathcal{Y}\to\mathcal{C}
\end{equation}
that recovers the transmitted codeword from the received sequence.
Under the $d$-deletion-correcting constraint \eqref{deletion_correction}, perfect decoding is achievable: for any transmitted codeword $u\in\mathcal{C}$ and any $y\in\mathcal{D}_d(u)$,
\begin{equation}
\label{perfect_decoding}
\phi(y) = u.
\end{equation}
The constraint \eqref{deletion_correction} ensures that a unique decoder exists by construction: each received sequence $y \in \mathcal{Y}$ corresponds to exactly one codeword in $\mathcal{C}$. We present an efficient implementation of such a decoder in Section~\ref{sec:decoding}.

Our objectives are twofold: first, to find the largest $d$-deletion-correcting codebook
\begin{equation}
\label{opt_problem}
\max_{\mathcal{C}\subseteq\mathcal{S}} \; |\mathcal{C}|
\quad\text{subject to}\quad
\mathcal{D}_d(u)\cap\mathcal{D}_d(v)=\varnothing\;\; \forall\; u\neq v\in\mathcal{C},
\end{equation}
and second, to construct an efficient decoder $\phi$ satisfying \eqref{perfect_decoding}.

\section{Solution}
\label{sec:solution}
We solve the optimization problem \eqref{opt_problem} by casting it as a Maximum Clique Problem (MCP) on an appropriately constructed graph. 
We first establish the necessary terminology and then present our construction and decoding algorithms.

\subsection{Preliminaries}

\subsubsection{Longest Common Subsequence}

Given two sequences \(x \in \{0,1\}^m\) and \(y \in \{0,1\}^n\), their \emph{longest common subsequence} (LCS) is the longest sequence obtainable from both \(x\) and \(y\) by deleting zero or more bits without reordering. 
Let \(\mathrm{LCS}(x,y)\) denote the length of an LCS of \(x\) and \(y\), computable in \(O(mn)\) time via dynamic programming.

The \emph{LCS distance} between \(x\) and \(y\) is defined as
\begin{equation}
d_{\mathrm{LCS}}(x,y) = \frac{|x| + |y| - 2 \cdot \mathrm{LCS}(x,y)}{2}.
\label{eq:lcs_distance}
\end{equation}
For sequences of equal length \(n\), this simplifies to
\begin{equation}
\label{eq:lcs_condition}
d_{\mathrm{LCS}}(x,y) = n - \mathrm{LCS}(x,y).
\end{equation}

\subsubsection{Cliques and the Maximum Clique Problem}

A graph \(G = (V,E)\) consists of a vertex set \(V\) and edge set \(E \subseteq V \times V\). 
A \emph{clique} is a subset \(C \subseteq V\) such that every pair of distinct vertices in \(C\) is connected:
\[
\forall\  u,v \in C, \; u \neq v \implies (u,v) \in E.
\]
The \emph{clique number} \(\omega(G)\) is the size of a largest clique in \(G\):
\[
\omega(G) = \max\{\,|C| : C \subseteq V,\ C\ \text{is a clique}\,\}.
\]
A clique \(C^*\) with \(|C^*|=\omega(G)\) is called a \emph{maximum clique}.
The Maximum Clique Problem (MCP) seeks to find such a clique:
\[
\text{find } C^* \subseteq V \quad\text{such that}\quad |C^*|=\omega(G).
\]

\subsection{LCS-Based Characterization of Deletion Correction}

The deletion-correcting constraint~\eqref{deletion_correction} requires that 
\[
\mathcal{D}_d(u) \cap \mathcal{D}_d(v) = \varnothing
\]
for all distinct codewords \(u, v \in \mathcal{C}\).
We now characterize this condition in terms of the longest common subsequence (LCS) distance.

\begin{remark}
\label{constraint_remark}
As established in~\cite{lcs_consrtaint_proof}, the deletion-correcting constraint is equivalent to an LCS-based distance constraint.  
Specifically, a codebook \(\mathcal{C} \subseteq \mathcal{S}\) satisfies~\eqref{deletion_correction} if and only if, for every pair of distinct codewords \(u, v \in \mathcal{C}\),
\begin{equation}
    \mathrm{LCS}(u, v) \le n - (d + 1),
    \label{eq:pairwise_lcs_constraint}
\end{equation}
or equivalently (from~\eqref{eq:lcs_condition}),
\begin{equation}
\label{eq:distance_lcs_condition}
    d_{\mathrm{LCS}}(u, v) \ge d + 1,
\end{equation}
where \(d_{\mathrm{LCS}}(u, v)\) denotes the LCS distance between \(u\) and \(v\).
\end{remark}

\subsection{Construction}
\label{sec:construction}
We now translate the deletion-correcting code construction problem into a MCP and describe our solution approach.

\subsubsection{Reduction to Maximum Clique Problem}

Considering equation~\eqref{eq:pairwise_lcs_constraint}, we construct an undirected graph \( G = (V, E) \) as follows:
\[
V = \{0,1\}^n, \quad (u,v) \in E \iff \mathrm{LCS}(u,v) \le n - d - 1,
\]
or equivalently, using~\eqref{eq:distance_lcs_condition},
\[
(u,v) \in E \iff d_{\mathrm{LCS}}(u,v) \ge d + 1.
\]
Each vertex represents a binary sequence of length \( n \), and an edge connects two sequences if they satisfy the pairwise deletion-correcting constraint. 
By Remark~\ref{constraint_remark}, any clique in \(G\) corresponds to a valid $d$-deletion-correcting codebook, and a maximum clique yields a codebook of maximum size.

\subsubsection{LCS Computation}

Computing the LCS length between two sequences is performed via standard dynamic programming in $O(n^2)$ time for equal-length sequences, as shown in Algorithm~\ref{alg:LCS}.

\begin{algorithm}[t]
\vspace{0.2cm}
\caption{LCS Length Computation}
\label{alg:LCS}
\begin{algorithmic}[1]
\Require Binary sequences \(x, y\) of lengths \(m, n\)
\Ensure \( \mathrm{LCS}(x, y) \)
\State Initialize matrix \( L[0 \ldots m][0 \ldots n] \) with zeros
\For{ \( i = 1 \) to \( m \) }
    \For{ \( j = 1 \) to \( n \) }
        \If{ \( x[i - 1] = y[j - 1] \) }
            \State \( L[i][j] \leftarrow L[i - 1][j - 1] + 1 \)
        \Else
            \State \( L[i][j] \leftarrow \max(L[i - 1][j], L[i][j - 1]) \)
        \EndIf
    \EndFor
\EndFor
\State \Return \( L[m][n] \)
\end{algorithmic}
\end{algorithm}

\subsubsection{Graph Construction and Clique Finding}

Algorithm~\ref{alg:MCP} constructs the graph \( G \) and invokes a clique-finding algorithm to obtain a large codebook.

\begin{algorithm}[t]
\vspace{0.2cm}

\caption{Codebook Construction via Maximum Clique}
\label{alg:MCP}
\begin{algorithmic}[1]
\Require \( n \) (codeword length), \( d \) (deletion correction capability)
\Ensure Large feasible codebook \( \mathcal{C} \)
\State \( \text{threshold} \leftarrow d + 1 \)
\State \( V \leftarrow \{0, 1\}^n \)
\State Initialize empty graph \( G = (V, \emptyset) \)
\For{ each unordered pair \( \{u, v\} \subset V \), \( u \neq v \) }
    \State Compute \( d_{\mathrm{LCS}}(u, v) \) using Algorithm~\ref{alg:LCS}
    \If{ \( d_{\mathrm{LCS}}(u, v) \ge \text{threshold} \) }
        \State Add edge \( (u, v) \) to \( E \)
    \EndIf
\EndFor
\State \( \mathcal{C} \leftarrow \text{FindMaxClique}(G) \)
\State \Return \( \mathcal{C} \)
\end{algorithmic}
\end{algorithm}

\subsubsection{Penalty-Guided Clique Search}
For small instances (\(n \leq 8\)), we can use the Bron--Kerbosch algorithm~\cite{Bron}
to compute an optimal maximum clique. For larger instances where exact enumeration becomes infeasible, we employ PGCS, a lightweight stochastic heuristic for finding large cliques in deletion-compatibility graphs. PGCS is inspired by the dynamic penalty mechanism of Dynamic Local Search (DLS)~\cite{DLS}.

The PGCS heuristic incrementally constructs a clique using greedy expansion
steps whenever feasible. When no further expansion is possible, the algorithm
performs a penalty-guided perturbation by removing a high-penalty vertex or
restarting from a randomly selected vertex, enabling continued exploration
while maintaining low per-iteration complexity. A dynamic penalty mechanism is used to discourage repeated selection of the same
vertices. Each vertex is assigned a penalty that is increased whenever it
appears in the current clique and gradually decayed over time. Candidate
vertices with lower penalties are preferred during clique expansion, promoting
diversification and reducing cycling behavior. The penalty increase and decay
parameters are empirically tuned to balance exploration and exploitation. The algorithm terminates when a prescribed time budget is reached and returns
the largest clique encountered during the search.

\begin{algorithm}[t]
\caption{Penalty-Guided Clique Search (PGCS) for Maximum Clique}
\label{alg:dls-mc}
\begin{algorithmic}[1]
\Require Graph $G=(V,E)$, penalty and time parameters: $\mathrm{penInc}, \mathrm{penDec}$, $\mathrm{maxTime}$
\Ensure Largest clique $C^*$ found within $\mathrm{maxTime}$
\State Initialize $\mathrm{penalty}(v)\gets 0$ for all $v\in V$
\State $C\gets\{\mathrm{RandomChoice}(V)\}$, $C^*\gets C$
\State $t_{\mathrm{start}}\gets$ current time
\While{current time $- t_{\mathrm{start}} < \mathrm{maxTime}$}
    \State $\mathrm{Cand}\gets\{v\in V\setminus C : (v,u)\in E,\ \forall u\in C\}$
    \If{$\mathrm{Cand}\neq\emptyset$}
        \State $v\gets\arg\min_{v\in\mathrm{Cand}}\mathrm{penalty}(v)$
        \State $C\gets C\cup\{v\}$
    \Else
        \If{$|C|>1$}
            \State $v\gets\arg\max_{v\in C}\mathrm{penalty}(v)$
            \State $C\gets C\setminus\{v\}$
        \Else
            \State $C\gets\{\mathrm{RandomChoice}(V)\}$
        \EndIf
    \EndIf
    \ForAll{$v\in C$}
        \State $\mathrm{penalty}(v)\gets\mathrm{penalty}(v)+\mathrm{penInc}$
    \EndFor
    \ForAll{$v\in V$}
        \State $\mathrm{penalty}(v)\gets\mathrm{penalty}(v)\cdot\mathrm{penDec}$
    \EndFor
    \If{$|C|>|C^*|$}
        \State $C^*\gets C$
    \EndIf
\EndWhile
\State \Return $C^*$
\end{algorithmic}
\end{algorithm}

\section{Decoding}
\label{sec:decoding}
We now describe the decoding algorithm that recovers the transmitted codeword from a deletion-corrupted received sequence. Throughout this section, we assume that codeword boundaries in the received bitstream are known (i.e., the receiver can segment the corrupted stream into individual received sequences). This assumption is standard in coded systems~\cite{segmented_decoding} and can be achieved through framing, marker sequences~\cite{Marker_code_marker, DNA_sync}, or synchronization protocols.

\begin{algorithm}[t]
\caption{LCS-Based Decoding}
\label{alg:lcsdecode}
\begin{algorithmic}[1]
\Require Codebook $\mathcal{C}$ with precomputed $(w(c), z(c))$ for each $c \in \mathcal{C}$, received sequence $y$
\Ensure Decoded codeword $\hat{c}$
\State Compute $w_y \gets w(y)$, $z_y \gets z(y)$
\State $\mathcal{C}' \gets \{c \in \mathcal{C} : w(c) \ge w_y \text{ and } z(c) \ge z_y\}$
\State $\hat{c} \gets \text{null}$, $S_{\max} \gets 0$
\For{$c \in \mathcal{C}'$}
    \State $S \gets \mathrm{LCS}(c, y)$ using Algorithm~\ref{alg:LCS}
    \If{$S = |y|$}
        \State \Return $c$ \Comment{Early exit: unique match found}
    \ElsIf{$S > S_{\max}$}
        \State $S_{\max} \gets S$, $\hat{c} \gets c$
    \EndIf
\EndFor
\State \Return $\hat{c}$
\end{algorithmic}
\end{algorithm}

\subsection{LCS-Based Decoder}

Our decoder is the following: given a received sequence $y \in \mathcal{Y}$, it selects the codeword with the highest LCS score:
\begin{equation}
\phi(y) = \arg\max_{u \in \mathcal{C}} \mathrm{LCS}(u, y).
\label{eq:lcs_decoder}
\end{equation}
This approach is similar to reconstruction methods used in DNA storage systems~\cite{DNA_sequence}.

\subsection{Decoder Correctness}
The following lemma formalizes the correctness of LCS-based decoding under the deletion-correcting constraint; it is included for completeness and to support the subsequent complexity analysis.

\begin{lemma}
\label{lemma:decoder_correctness}
Let $\mathcal{C}$ be a $d$-deletion-correcting codebook satisfying equation~\eqref{eq:pairwise_lcs_constraint}. 
The LCS-based decoder \eqref{eq:lcs_decoder} uniquely recovers the transmitted codeword from any received sequence $y \in \mathcal{Y}$.
\end{lemma}

\begin{proof}
Let $x \in \mathcal{C}$ be the transmitted codeword, and suppose the received sequence $y$ results from $t \le d$ deletions. 
Then $|y| = n - t$ and $y$ is a subsequence of $x$, giving
\begin{equation}
\mathrm{LCS}(x,y) = |y| = n - t.
\end{equation}

For any other codeword $z \in \mathcal{C}$ with $z \neq x$, by~\eqref{eq:pairwise_lcs_constraint},
\begin{equation}
\mathrm{LCS}(x,z) \le n - d - 1.
\end{equation}
Now consider $\mathrm{LCS}(z,y)$. Since $y$ is obtained from $x$ by deletions,
we have $|y| \leq |x|$. Let $s$ denote the longest common subsequence of $z$ and $y$,
so that $|s|=\mathrm{LCS}(z,y)$. By definition, $|s| \leq |z|$ and $|s| \leq |y|$.
Since $|y| \leq |x|$, the transitivity of the subsequence relation implies
$|s| \leq |x|$ as well. Hence, $s$ is a common subsequence of $x$ and $z$, and
therefore
\begin{equation}
\mathrm{LCS}(z,y) = |s| \le \mathrm{LCS}(x,z) \le n-d-1.
\end{equation}
Since $t \le d$, it follows that
\begin{equation}
\mathrm{LCS}(z,y) \le n-d-1 < n-d \le n-t = |y| = \mathrm{LCS}(x,y).
\end{equation}
Therefore, $x$ is the unique codeword achieving the maximum LCS score with $y$, and $\phi(y) = x$.
\end{proof}

\subsection{Implementation and Complexity}

The decoder evaluates $\mathrm{LCS}(u,y)$ for each codeword $u \in \mathcal{C}$ using Algorithm~\ref{alg:LCS}, requiring $O(|\mathcal{C}| n^2)$ time and $O(n^2)$ space. 
This is significantly more efficient than enumerating all $\binom{n}{d}$ deletion patterns to reconstruct the original codeword.

\subsubsection{Optimizations}
We introduce two optimizations that preserve correctness while improving the practical efficiency of LCS-based decoding, reducing the computational cost below the baseline $O(|\mathcal{C}|n^2)$ complexity.

\paragraph{Symbol-count filtering.}
Since deletions can only reduce symbol counts, any valid source codeword $c$ must satisfy
\[
w(c) \ge w(y), \quad z(c) \ge z(y),
\]
where $w(\cdot)$ and $z(\cdot)$ denote the counts of ones and zeros, respectively.
This lightweight filtering step costs $O(|\mathcal{C}| n)$ and typically eliminates a large fraction of candidates before performing LCS computation.

\paragraph{Early termination.}
If a codeword $c$ achieves $\mathrm{LCS}(c, y) = |y|$, then $y$ is a subsequence of $c$.
By Lemma~\ref{lemma:decoder_correctness}, this codeword is unique, and decoding can terminate immediately.
Algorithm~\ref{alg:lcsdecode} combines these optimizations.

\begin{table*}[!t]
\vspace{0.2cm}
\centering
\caption{Comparison of Double-Deletion ($d=2$) Correcting Codebook Cardinalities~\cite{Compare,coloring} (Leven. LB = Levenshtein lower bound)}
\label{tab:double_del_results}
\begin{tabular}{|c|c|c|c|c|c|c|c|}
\hline
\textbf{n} & \textbf{Leven. LB} & \textbf{Helberg} & \textbf{Swart} & \textbf{Optimal} & \textbf{Coloring} & \textbf{Min-Deg} & \textbf{PGCS} \\
\hline
3  &0  & 2 & - & 2 & 2 & 2 & 2  \\
\hline
4  &0  & 2 & 2 & 2 & 2 & 2 & 2  \\
\hline
5  &0  & 2 & 2 & 2 & 2 & 2 & 2  \\
\hline
6  &0  & 3 & 4 & 4 & 4 & 4 & 4  \\
\hline
7  &0  & 4 & 5 & 5 & 5 & 5 & 5  \\
\hline
8  &1  & 5 & 7 & 7 & 6 & 7 & 7  \\
\hline
9  &1  & 6 & 10 & 11 & 8 & 10 & 11 \\
\hline
10 &1 & 8 & 14 & 16 & 12 & 15 & 16  \\
\hline
11 &2 & 9 & 20 & 24 & 17 & 21 & 23 \\
\hline
12 &3 & 11 & 29 & - & 27 & 32 & 35 \\
\hline
13 &4 & 15 & - & - & 40 & 49 & 53  \\
\hline
14 &6 & 18 & - & - & 60 & 78 & 81 \\
\hline
\end{tabular}
\end{table*}
\label{sec:results}

\begin{table*}[!t]
\centering
\caption{Comparison of Codebook Cardinalities for $d \in \{1, 3, 4\}$~\cite{Compare}}
\label{tab:multi_del_results}
\begin{tabular}{|c||c|c|c||c|c|c||c|c|c|}
\hline
\multirow{2}{*}{$n$}
 & \multicolumn{3}{c||}{$d=1$} & \multicolumn{3}{c||}{$d=3$} & \multicolumn{3}{c|}{$d=4$} \\
\cline{2-10}
 & \textbf{VT} & \textbf{Min-Deg} & \textbf{PGCS} & \textbf{Helberg} & \textbf{Min-Deg} & \textbf{PGCS} &  \textbf{Helberg} & \textbf{Min-Deg} &\textbf{PGCS}  \\
\hline
2 &2  &2  &2  &-  &-  &-  &-  &- &-\\
\hline
3 &2  &2  &2  &-  &-  &-  &-  &- &-\\
\hline
4 &4  &4  &4  &2  &2  &2  &-  &- &-\\
\hline
5 &6  &6  &6  &2  &2  &2  &2  &2 &2\\
\hline
6 &10 &10 &10  &2  &2  &2  &2 &2 &2\\
\hline
7 &16 &15 &16 &2 &2 &2 &2 &2 &2  \\
\hline
8 &30 &26 &28 &3 &4 &4 &2 &2  &2  \\
\hline
9 &52 &43 &45 &4 &5 &5 &2 &2 &2  \\
\hline
10 &94 &76 &80 &4 &6 &6 &3 &4 &4  \\
\hline
11 &170 &130 &141 &5 &8 &8 &4 &5 &5  \\
\hline
\end{tabular}
\end{table*}

\section{Results and Discussion}

\subsection{Codebook Construction Results}
We present experimental results comparing the proposed PGCS-based approach with existing methods from the literature.

Table~\ref{tab:double_del_results} shows codebook sizes for double-deletion correction ($d=2$) and codeword lengths $n=3$ to $14$. 
The proposed PGCS-based construction consistently achieves larger codebooks than Helberg's construction~\cite{Helberg}, Swart's heuristic~\cite{Swart}, the minimum-degree algorithm~\cite{Compare}, and the vertex-coloring method~\cite{coloring}.

Table~\ref{tab:multi_del_results} presents results for other deletion levels $d \in \{1, 3, 4\}$. 
For $d=1$, we compare against VT codes~\cite{VT_original}; for $d \geq 3$, against Helberg codes~\cite{Helberg}.
The minimum-degree algorithm~\cite{Compare} serves as a baseline across all deletion levels.

\subsubsection{Implementation Details and Reproducibility}
\label{sec:implementation}

All experiments were conducted on a machine with an Intel Core i5--1335U CPU (13th
Gen) and 16~GB RAM.  Since the proposed PGCS algorithm is stochastic,
each $(n,d)$ instance was evaluated over $10$ independent runs with different
random seeds, each with a runtime budget of $120$ seconds. Tables~\ref{tab:double_del_results}
and~\ref{tab:multi_del_results} report the maximum codebook size obtained across
these runs, corresponding to the largest feasible codebook discovered within the
fixed computational budget. 
We observed that variability across runs was empirically small for most instances,
with the best and median solutions typically coinciding.

The PGCS penalty parameters $(\mathrm{penInc}, \mathrm{penDec})$, which govern the
dynamic penalty mechanism used for vertex selection, were chosen from a small set
of candidate values based on preliminary experiments and then fixed for all runs
of a given $(n,d)$ configuration. The parameter $\mathrm{penInc}$ controls the
magnitude of the penalty assigned to vertices participating in the current
clique, while $\mathrm{penDec}$ determines how penalty values evolve over time.
Together, these parameters regulate the influence of accumulated penalties on
the vertex selection process.

Our experiments indicate that suitable parameter choices correlate with the
density of the underlying clique graph induced by the deletion constraint. For
moderately sparse graphs (e.g., $d=2$ with $n \leq 14$), the parameter setting
$(\mathrm{penInc}, \mathrm{penDec}) = (0.95, 0.8)$ consistently yielded strong
performance across all tested instances. For denser graphs (e.g., $d=1$),
alternative parameter settings were empirically observed to perform better; for
example, $(\mathrm{penInc}, \mathrm{penDec}) = (0.4, 0.95)$ was used for
$(n,d)=(11,1)$. These values were selected empirically and are reported for
reproducibility rather than to suggest a universal parameter-selection rule.

We adopt PGCS rather than a full Dynamic Local Search (DLS) implementation
because the clique graphs arising from deletion-correcting code construction
exhibit substantial structure and are significantly smaller than the benchmark
graphs typically targeted by full DLS solvers. In this setting, the streamlined
penalty-guided add--remove strategy of PGCS achieves competitive or superior
solutions at substantially lower computational overhead, making it well suited
for repeated evaluation across multiple $(n,d)$ configurations.

\subsubsection{Scalability Limitations}
In our implementation, we explicitly construct and store the compatibility graph $G = (V,E)$ with $|V| = 2^n$ vertices and $|E| = O(4^n)$ edges.
The graph construction dominates both time and memory complexity, scaling exponentially in $n$ regardless of the clique-finding algorithm used.
Although prior works report results for $n \le 20$~\cite{Compare}, we restrict our experiments to $n \le 14$ due to these computational constraints.
Within this range, PGCS consistently yields larger or equal-sized codebooks
compared to existing methods, demonstrating its effectiveness as a practical
maximum clique solver for this problem.

\subsection{Decoding Performance}

We evaluate the practical decoding complexity of the proposed LCS-based decoder through Monte Carlo simulations under segmented reception, where codeword boundaries are known.

\subsubsection{Illustrative Example.}
We first illustrate the decoding process using the representative configuration $(n,d) = (13,2)$. Consider a transmitted codeword $\texttt{1000111110110}$ of length $n=13$, containing $8$ ones and $5$ zeros. After two random deletions, the received sequence becomes $\texttt{10011111010}$, which contains $7$ ones and $4$ zeros. In the first stage of decoding, symbol-count filtering is applied: only those codewords whose counts of ones and zeros are at least as large as those in the received sequence are retained as candidates. In this example, $31$ out of $50$ codewords (obtained in a specific run) fail this necessary condition and are discarded, leaving $19$ candidates. Separately, a Monte Carlo simulation yields an average search-space reduction of over $60\%$ for $d=2$, where the number of deletions takes values in $\{0,1,2\}$, and, conditioned on the number of deletions, all possible deletion patterns are assumed to be uniformly distributed.

Next, LCS-based decoding is performed over the remaining candidates. Since the transmitted codeword achieves an LCS score equal to the received sequence length ($\mathrm{LCS}=11=|y|$), it is immediately identified as the correct codeword via the early-termination criterion. This example illustrates how the combination of lightweight pre-filtering and early termination substantially reduces the number of LCS computations without affecting decoding correctness.

\section{Conclusion}
\label{sec:conclusion}
We investigated the construction of deletion-correcting codes through a maximum
clique formulation and demonstrated that the proposed
\emph{Penalty-Guided Clique Search} heuristic yields consistently improved
results compared to simpler polynomial-time graph-based construction methods.
For finite block lengths $n \leq 14$ and deletion parameters $d \in \{1,2,3\}$,
the PGCS-based approach produces larger codebooks than previously reported
clique-based heuristics, including minimum-degree and coloring methods, as well
as established finite-length constructions such as Helberg codes and related
variants. In several instances, the resulting codebooks match known optimal
sizes for the corresponding $(n,d)$ parameters. We also proposed an optimized LCS-based decoder that significantly reduces
practical decoding complexity by incorporating lightweight symbol-count
filtering and early termination, while preserving exact decoding guarantees
under segmented reception.

The primary computational bottleneck of the proposed framework lies in the
construction of the compatibility graph, which requires evaluating pairwise LCS
distances among $2^n$ candidate sequences and storing up to $O(2^{2n})$ edges.
This currently limits practical applicability to moderate block lengths.
Future research directions include exploring alternative maximum clique
solvers and penalty-based local search heuristics inspired by DLS to further improve solution quality and scalability.
\balance
\bibliographystyle{IEEEtran}
\bibliography{references}

\begin{thebibliography}{10}
\providecommand{\url}[1]{#1}
\csname url@samestyle\endcsname
\providecommand{\newblock}{\relax}
\providecommand{\bibinfo}[2]{#2}
\providecommand{\BIBentrySTDinterwordspacing}{\spaceskip=0pt\relax}
\providecommand{\BIBentryALTinterwordstretchfactor}{4}
\providecommand{\BIBentryALTinterwordspacing}{\spaceskip=\fontdimen2\font plus
\BIBentryALTinterwordstretchfactor\fontdimen3\font minus \fontdimen4\font\relax}
\providecommand{\BIBforeignlanguage}[2]{{%
\expandafter\ifx\csname l@#1\endcsname\relax
\typeout{** WARNING: IEEEtran.bst: No hyphenation pattern has been}%
\typeout{** loaded for the language `#1'. Using the pattern for}%
\typeout{** the default language instead.}%
\else
\language=\csname l@#1\endcsname
\fi
#2}}
\providecommand{\BIBdecl}{\relax}
\BIBdecl

\bibitem{symbol_sync}
\BIBentryALTinterwordspacing
H.~Mercier, ``{Communication over channels with symbol synchronization errors},'' Ph.D. dissertation, University of British Columbia, 2008. [Online]. Available: \url{https://open.library.ubc.ca/collections/ubctheses/24/items/1.0066430}
\BIBentrySTDinterwordspacing

\bibitem{DNA_storage}
Z.~Yan, C.~Liang, and H.~Wu, ``{Upper and Lower Bounds on the Capacity of the DNA-Based Storage Channel},'' \emph{IEEE Commun. Lett.}, vol.~26, no.~11, pp. 2586--2590, Nov 2022.

\bibitem{Compare}
F.~Khajouei, M.~Zolghadr, and N.~Kiyavash, ``{An algorithmic approach for finding deletion correcting codes},'' in \emph{IEEE ITW}, Oct 2011, pp. 25--29.

\bibitem{VT_original}
R.~R. Varshamov and G.~M. Tenengol'ts, ``Code correcting single asymmetric errors,'' \emph{Avtomatika i Telemekhanika}, vol.~26, no.~2, pp. 288--292, 1965.

\bibitem{VT_code}
M.~Abroshan, R.~Venkataramanan, and A.~G.~I. Fabregas, ``{Efficient Systematic Encoding of Non-binary VT Codes},'' in \emph{IEEE ISIT}, June 2018, pp. 91--95.

\bibitem{Helberg}
A.~Helberg and H.~Ferreira, ``{On multiple insertion/deletion correcting codes},'' \emph{IEEE Trans. Inf. Theory}, vol.~48, no.~1, pp. 305--308, Jan 2002.

\bibitem{LeNyugen}
T.~A. Le and H.~D. Nguyen, ``New multiple insertion/deletion correcting codes for non-binary alphabets,'' \emph{IEEE Trans. Inf. Theory}, vol.~62, no.~5, pp. 2682--2693, May 2016.

\bibitem{coloring}
D.~Cullina, A.~A. Kulkarni, and N.~Kiyavash, ``{A coloring approach to constructing deletion correcting codes from constant weight subgraphs},'' in \emph{IEEE ISIT}, July 2012, pp. 513--517.

\bibitem{DLS}
W.~Pullan and H.~H. Hoos, ``{Dynamic local search for the maximum clique problem},'' \emph{J. Artif. Int. Res.}, vol.~25, no.~1, p. 159–185, Feb. 2006.

\bibitem{DNA_sequence}
\BIBentryALTinterwordspacing
O.~Sabary, A.~Yucovich, G.~Shapira, and E.~Yaakobi, ``{Reconstruction Algorithms for DNA-Storage Systems},'' \emph{bioRxiv}, 2020. [Online]. Available: \url{https://www.biorxiv.org/content/early/2020/09/17/2020.09.16.300186}
\BIBentrySTDinterwordspacing

\bibitem{segmented_decoding}
Z.~Liu and M.~Mitzenmacher, ``{Codes for Deletion and Insertion Channels With Segmented Errors},'' \emph{IEEE Trans. Inf. Theory}, vol.~56, no.~1, pp. 224--232, Jan 2010.

\bibitem{DNA_sync}
Z.~Yan, C.~Liang, and H.~Wu, ``{A Segmented-Edit Error-Correcting Code With Re-Synchronization Function for DNA-Based Storage Systems},'' \emph{IEEE Trans. Emerg. Top. Comput.}, vol.~11, no.~3, pp. 605--618, July 2023.

\bibitem{Marker_code_marker}
Z.~Li, X.~He, and X.~Tang, ``Marker+codeword+marker: A coding structure for segmented single-insdel/-edit channels,'' \emph{IEEE Trans. Commun.}, pp. 1--1, 2025.

\bibitem{DNA_deletion}
J.~Haghighat and T.~M. Duman, ``{Half-Marker Codes for Deletion Channels With Applications in DNA Storage},'' \emph{IEEE Commun. Lett.}, vol.~29, no.~7, pp. 1639--1643, July 2025.

\bibitem{lcs_consrtaint_proof}
\BIBentryALTinterwordspacing
R.~Con, Z.~Guo, R.~Li, and Z.~Zhang, ``{Random Reed-Solomon Codes Achieve the Half-Singleton Bound for Insertions and Deletions over Linear-Sized Alphabets},'' 2024. [Online]. Available: \url{https://arxiv.org/abs/2407.07299}
\BIBentrySTDinterwordspacing

\bibitem{Bron}
\BIBentryALTinterwordspacing
C.~Bron and J.~Kerbosch, ``{Algorithm 457: finding all cliques of an undirected graph},'' \emph{Commun. ACM}, vol.~16, no.~9, p. 575–577, Sep. 1973. [Online]. Available: \url{https://doi.org/10.1145/362342.362367}
\BIBentrySTDinterwordspacing

\bibitem{Swart}
T.~Swart and H.~Ferreira, ``{A note on double insertion/deletion correcting codes},'' \emph{IEEE Trans. Inf. Theory}, vol.~49, no.~1, pp. 269--273, Jan 2003.

\end{thebibliography}

\end{document}